\documentclass{article}

 \usepackage[preprint]{neurips_2019}

\usepackage[utf8]{inputenc} 
\usepackage[T1]{fontenc}    
\usepackage{hyperref}       
\usepackage{url}            
\usepackage{booktabs}       
\usepackage{amsfonts}       
\usepackage{nicefrac}       
\usepackage{microtype}      
\usepackage{amsmath}
\usepackage{amsthm}
\usepackage{diagbox}
\usepackage[linesnumbered,ruled]{algorithm2e}

\newtheorem{theorem}{Theorem}

\title{Quantification of the Leakage in Federated Learning}

\author{%
  Zhaorui Li\thanks{work done while at Alibaba security} \\
  Shanghai Jiao Tong University \\
  \texttt{lizhaorui@sjtu.edu.cn} \\
  \And
  Zhicong Huang \\
  Alibaba security \\
  \texttt{zhicong.hzc@alibaba-inc.com} \\
  \AND
  Chaochao Chen \\
  Ant Financial \\
  \texttt{chaochao.ccc@alipay.com} \\
  \And
  Cheng Hong \\
  Alibaba security \\
  \texttt{vince.hc@alibaba-inc.com} \\
}

\begin{document}

\maketitle

\begin{abstract}
  With the growing emphasis on users' privacy, federated learning has become more and more popular. Many architectures have been raised for a better security. Most architecture work on the assumption that data's gradient could not leak information. However, some work, recently, has shown such gradients may lead to leakage of the training data. In this paper, we discuss the leakage based on a federated approximated logistic regression model and show that such gradient's leakage could leak the complete training data if all elements of the inputs are either 0 or 1.
\end{abstract}

\section{Introduction} \label{introduction}
Federated learning, where multiple parties construct a joint model with each party's own data, has become more and more popular with the increase of the emphasis on privacy. Many architectures have been proposed (e.g.\citet{mcmahan2016communication} and \citet{konevcny2016federated}) for a tradeoff between effeciency and security. However some work have analyzed the leakage in federated learning (e.g. \citet{hitaj2017deep}). Hence many defense work have been proposed (e.g. \citet{bonawitz2017practical} and \citet{geyer2017differentially}). One of the defenses is to use additive homomorphic encryption to encrypt model parameters for secure aggregation during the update procedure (\citet{aono2017privacy}).

However this technique might not mitigate the leakage in federated learning. \citet{melis2018exploiting} has shown an honest-but-curious participant could obtain the gradient computed by others through the difference of the global joint model and thus can infer unintended feature of the training data. \citet{ligeng2019leakage} `steal's the training data pixel-wise from gradients. But these methods can only work on complicated networks (e.g. Convolutional Neural Network) due to its huge parameters and do not apply to simple models like the logistic regression model.

For logistic regression model, the loss function cannot be directly encrypted by additive homomorphic encryption due to the imcompatiablity of sigmoid function and additive homomorphic encryption. Thus many approximation work (e.g. \citet{aono2016scalable}) are proposed, of which the goals are to increase accuracy as well as approxiamate the sigmoid function for additive homomorphic encryption.

In this paper, we discuss the leakage of the federated approximated logistic regression model in the case where all elements of the input are binary, which is a widely used encoding method in gene data analysis (e.g. \citet{uhlerop2013privacy}) and risk analysis. The loss function of the model is approximated as the way raised by \citet{aono2016scalable}. We will show the training data can completely be inferred by an honest-but-curious participant.

\section{Related Work}
\citet{aono2017privacy} has shown part of the training data is leaked in collaborative learning in the situation that one batch only contains a single data. \citet{melis2018exploiting} proved that the features of the training data which are unrelated with the model target leak through the update procedure. \citet{ligeng2019leakage} use an optimization method to infer the whole training data based on the gradient leaked from the update procedure. However this optimization method does not apply to the approximated logistic regression model since gradients computed based on different batches could be same. In addition, none of these works consider the case where the gradient is computed based on multiple batches.

\section{Leakage Analysis}

\subsection{Preliminary} \label{Background}
Consider the two-party case of federated learning, where two parties (\textit{Alice} and \textit{Bob}) share the same logistic regression model. The parameters of the federated model is $\theta$ and the feature num of inputs is $d$. Consider the case where all elements of the input data are either 0 or 1 which has practial application in domains like gene analysis (e.g. \citet{uhlerop2013privacy}) and risk analysis.

The way of training is the same as the horizontal case in \citet{yang2019federated}: In each iteration, (1) \textit{Alice} (\textit{Bob}) computes local training gradients $\Delta_n^A$ ($\Delta_n^B$) by Stochastic Gradient Descent (SGD) and sends $\Delta_n^A$ ($\Delta_n^B$) to the server. (2) The server uses secure aggregation raised by \citet{aono2017privacy} to aggregate $\Delta\theta_n=\frac{1}{2}(\Delta_n^A+\Delta_n^B)$ (3) The server sends the aggregated gradient $\Delta_n$ back to \textit{Alice} and \textit{Bob}. (4) \textit{Alice} (\textit{Bob}) updates his local model by $\theta_{n+1}=\theta_n-\Delta\theta_n$.

The loss function of this model is $l=log2-\frac{1}{2}(\theta^Tx)y+\frac{1}{8}(\theta^Tx)^2$ the same as the approximated way of \citet{aono2016scalable} for the secure aggregation in the third step of training ($x$ is a piece of training data and $y$ is its label). The gradient of the loss function in the matrix form is $\nabla=\frac{1}{4}X^TX\theta-\frac{1}{2}X^TY$.

Notice that $\Delta_n^B=2\Delta\theta_n-\Delta_n^A$. Suppose \textit{Alice} is an honest-but-curious participant and thus $\Delta_n^B$ is leaked to \textit{Alice}. We assume \textit{Bob} uses the same data for all iterations and \textit{Alice} can obtain as many pairs ($\Delta_n^B$, $\theta_{n-1}$) as possible. There are two common method to compute the training gradients in the first step of training: \textit{synchronized} and \textit{asynchronized}. In the following sections we will ommit the superscript $B$ and the subscript $n$ and discuss what \textit{Alice} can infer from the obtained set $\mathcal{S}=\{(\Delta$, $\theta)\}$ seperately with both methods.

\subsection{Synchronized} \label{OneBatch}
We first consider the \textit{synchronized} case, where \textit{Bob} calculates gradients based on his own data in current batch. Without loss of generality we assume \textit{Bob} only has one batch of data $X$ and the corresponding labels $Y$. Let $\alpha$ denote $X^T$ and $\beta$ denote $X^TY$ and thus $\Delta=\lambda\nabla=\lambda(\frac{1}{4}\alpha\theta-\frac{1}{2}\beta)$. Notice that $\Delta_{(i)}=\lambda(\frac{1}{4}\alpha_{(i)}\theta-\frac{1}{2}\beta_{(i)})$ (the subscript $(i)$ denote the $i$th row of a matrix) is a linear equation. Hence $\alpha_{(i)}$ and $\beta_{(i)}$ can be solved out if $\|\mathcal{S}\|\geq d$. Therefore $\alpha=X^TX$ and $\beta=X^TY$ are leaked. If the batch $X$ contains $m$ samples, \textit{Alice} can obtain:
\begin{equation}
  \label{alpha}
  \alpha_{ij}=\sum_{k=1}^mx_{ki}x_{kj} \quad i,j=1,2,\dots,d
\end{equation}
where $X=\{x_{ki}\}_{k=1,2,\cdots,m,i=1,2,\cdots,d}$ is \textit{Bob}'s private data. Futhermore, \textit{Alice} can change the knotty quadratic equation set to a linear programming and solve out $X$: 

According to Equation.\eqref{alpha}, $\alpha_{ii}=\sum_{k=1}^mx_{ki}^2$. Since $x_{ki}\in\{0, 1\}$, $x_{ki}=x_{ki}^2$. Thus we have $\alpha_{ii}=\sum_{k=1}^mx_{ki}$. For $\alpha_{ij},i\neq j$, let $\delta_{ijk}$ denote $x_{ki}x_{kj}$, we have $\alpha_{ij}=\sum_{k=1}^m\delta_{ijk}$. Given that if and only if $x_{ki}=x_{kj}=1$, $\delta_{ijk}=1$, otherwise $\delta_{ijk}=0$, we can rewrite this relation to a if-then constraint: if $\delta_{ijk}=1$ then $x_{ki}+x_{kj}=2\geq2$ and if $\delta_{ijk}=0$ then $x_{ki}+x_{kj}\leq1$.

Consider the first constraint: if $\delta_{ijk}=1$ then $x_{ki}+x_{kj}\geq2$. The inequality $-x_{ki}-x_{kj}+2\leq 2(1-\delta_{ijk})$ is equivalent to the constriant (When $\delta_{ijk}=1$, the inequality is the original constraint. When $\delta_{ijk}=0$, the inequality $-x_{ki}-x_{kj}+2\leq 2$ always holds). Hence $x_{ki}+x_{kj}\geq2\delta_{ijk}$. Similarly we can get the equivalent linear constraint form of the second constraint: $x_{ki}+x_{kj}-1\leq\delta_{ijk}$.

Hence \textit{Alice} can change Equation.\eqref{alpha} to the following linear constraints:

\begin{equation*}
  \begin{aligned}
    \sum_{k=1}^mx_{ki}=\alpha_{ii}&,\,i=1,2,\cdots,d\\
    \sum_{k=1}^m\delta_{ijk}=\alpha_{ij}&,\,i,j=1,2,\cdots,d,i\neq j\\
    x_{ki}+x_{kj}\geq2\delta_{ijk}&,\,i,j=1,2,\cdots,d,i\neq j,\,k=1,2,\cdots,m\\
    x_{ki}+x_{kj}-1\leq\delta_{ijk}&,\,i,j=1,2,\cdots,d,i\neq j,\,k=1,2,\cdots,m\\
    x_{ij},\delta_{ijk}\in\{0,1\}&,\,i,j=1,2,\cdots,d,\,k=1,2,\cdots,m
  \end{aligned}
\end{equation*} 

This linear programming problem can be easily solved by revised simplex method or inner point method as shown in Table.\ref{result}. Futhermore, even if the data has non-binary features, this method can separate them from binary features and reveal the latter.

\subsection{Asynchronized} \label{Multi Batches}
In this section we will discuss the leakage in \textit{asychronized} case, where \textit{Bob} uses serveral batches to calculate the gradients. For each batch $X_i$, \textit{Bob} calculates its gradients $\nabla_i$ based on the current local model $\theta_i$ and then updates the local model by $\theta_{i+1}=\theta_i-\lambda\nabla_i$. After all batches have been used \textit{Bob} pushes the difference between the current local model and the original global model $\Delta=\theta_1-\theta_{n+1}$ to the server. We first give the mathematical form of $\Delta$.

\begin{theorem} 
  \label{gamma}
  Suppose \textit{Bob} uses $n$ batches in \textit{asychronized} cases. Let $X_i$ denote the $i$th batch, $Y_i$ denote the corresponding labels, $\alpha_i$ denote $X_i^TX_i$ and $\beta_i$ denote $X_i^TY_i$. Therefore $\forall(\Delta,\theta)\in\mathcal{S},\Delta=\left[I-\prod_{i=1}^n(I-\frac{\lambda}{4}\alpha_i)\right]\theta-\frac{\lambda}{2}\Big(\sum_{i=1}^{n-1}\left[\prod_{j=i+1}^n(I-\frac{\lambda}{4}\alpha_j)\right]\beta_i+\beta_n\Big)$\footnote{$\prod$ denotes the continuous multiplication of matrices, i.e. $\prod_{i=1}^kM_i=M_kM_{k-1}\cdots M_1$} ($\lambda$ is the learning rate).
\end{theorem}
\begin{proof}
  Since for eack batch $X_i$, its gradient is computed as $\nabla_i=\frac{1}{4}\alpha_i\theta_i-\frac{1}{2}\beta_i$, the theorem is true when $n=1$. Now we assume that the theorem is true for any $n\leq k$ and any $n$ batches of data and consider the case $n=k+1$.

  Notice that $\Delta=\theta_1-\theta_{n+1}=\lambda\sum_{i=1}^n\nabla_i$. Let $\Delta_k$ denote $\lambda\sum_{i=1}^k\nabla_i$ and thus $\Delta_{k+1}=\Delta_k+\lambda\nabla_{k+1}$.

  Since $\nabla_{k+1}=\frac{1}{4}\alpha_{k+1}\theta_{k+1}-\frac{1}{2}\beta_{k+1}=\frac{1}{4}\alpha_{k+1}(\theta-\Delta_k)-\frac{1}{2}\beta_{k+1}$, we can obtain that 
  \begin{align*}
    \Delta_{k+1}&=\frac{\lambda}{4}\alpha_{k+1}\theta+(I-\frac{\lambda}{4}\alpha_{k+1})\Delta_k-\frac{\lambda}{2}\beta_{k+1}\\
    &=\frac{\lambda}{4}\alpha_{k+1}\theta+\left[I-\frac{\lambda}{4}\alpha_{k+1}-\prod_{i=1}^{k+1}(I-\frac{\lambda}{4}\alpha_i)\right]\theta\\
    &\quad-\frac{\lambda}{2}\Bigg(\sum_{i=1}^{k-1}\left[\prod_{j=i+1}^{k+1}(I-\frac{\lambda}{4}\alpha_j)\right]\beta_i+(I-\frac{\lambda}{4}\alpha_{k+1})\beta_k+\beta_{k+1}\Bigg)\\
    &=\left[I-\prod_{i=1}^{k+1}(I-\frac{\lambda}{4}\alpha_i)\right]\theta-\frac{\lambda}{2}\Bigg(\sum_{i=1}^{k}\left[\prod_{j=i+1}^{k+1}(I-\frac{\lambda}{4}\alpha_j)\right]\beta_i+\beta_{k+1}\Bigg)
  \end{align*}
  Hence the theorem is true for all $n\geq1$.
\end{proof}

By Theorem.\ref{gamma}, $\gamma=I-\prod_{i=1}^{k+1}(I-\frac{\lambda}{4}\alpha_i)$ and $\eta=\Big(\sum_{i=1}^{n-1}\left[\prod_{j=i+1}^n(I-\frac{\lambda}{4}\alpha_j)\right]\beta_i+\beta_n\Big)$ is leaked, as the same way of the leakage of $\alpha$ in Section.\ref{OneBatch}.
\begin{theorem}
  \label{myriad_solution}
  The solution of $\gamma=I-\prod_{i=1}^n(I-\frac{\lambda}{4}\alpha_i)$ is infinite under the constraint that $\forall i=1,2,\dots,n$, $\alpha_i$ is symmetrical\footnote{The constraint is because $\alpha_i=X_i^TX_i$}.
\end{theorem}
\begin{proof}
  When treating the equation as an equation set, we have $\frac{d(d+1)}{2}$ variables to determine $\alpha_i$ due to its symmetry. Therefore we have $n\times\frac{d(d+1)}{2}$ variables with $d^2$ equations. When $n\geq2$, $n\times\frac{d(d+1)}{2}\geq d(d+1)>d^2$. Thus the original equation has myriad solutions.
\end{proof}

According to Theorem.\ref{myriad_solution}, $\gamma$ leads to myriad possibilities of the batches matching the obtained $\Delta$. Notice the equation set $\eta$ contains $n*d$ variables with $d$ equations even when $\{\alpha_i\}_{i=1,2,\dots,n}$ is known, which means $\eta$ cannot help to reduce the possibilities. In this circumstance, it is unclear what further information beyond $\gamma$ and $\eta$ about the target's data could be inferred. It is interesting for future work to formally justify the leakage when infinite solutions are found, e.g., how much additional information is sufficient to reduce the solution space to a few plausible ones or even a single one.

\subsection{Multi-party}
In the \textit{synchronized} case where there are $k$ ($k\geq3$) parties in the federated learning, the global model is updated as $\theta_n=\frac{1}{k}\sum_{i=1}^k\theta_n^{(i)}$. Thus $-k\Delta\theta_n=\sum_{i=1}^k\Delta_n^{(i)}=\lambda\sum_{i=1}^k(\frac{1}{4}\alpha^{(i)}\theta_{n-1}-\frac{1}{2}\beta^{(i)})=\lambda\left[\frac{1}{4}(\sum_{i=1}^k\alpha^{(i)})\theta_{n-1}-\frac{1}{2}\sum_{i=1}^k\beta^{(i)}\right]$ (the superscript $(i)$ is used to denote the $i$th party).

Notice $P=\Bigl(\begin{smallmatrix} X^1 \\ \vdots \\ X^k \end{smallmatrix}\Bigr)$ is one solution of the equation $\sum_{i=1}^k\alpha^{(i)}=P^TP$. Therefore increasing the number of the parties in \textit{synchronized} case is equivalent to increasing the batch size. An honest-but-curious party can infer all other participants' data only not knowing their belongings.
  
However for the \textit{asychronized} case, $\sum_{i=1}^k\gamma^{(i)}=kI-\sum_{i=1}^k\prod_{j=1}^n(I-\frac{\lambda}{4}\alpha_j^{(i)})=I-\prod_{j=1}^n(I-\frac{\lambda}{4}P_j)$ has no simple solution for $\{P_j\}_{j=1,2,\cdots,n}$. Hence the increase of party num in asychronized case is unequal to the increase of batch size. We do not further analyze how to simplify the multi-party case to the two-party case since \textit{Alice} can only infer some constraints in the two-party case for now.

\section{Defense}
\subsection{Batch size}

\begin{table}[tbp]
  \centering
  \begin{tabular}{|c|c|c|c|c|}
    \hline
    \diagbox{$m$}{time (s)}{$d$} & 5 & 10 & 15 & 20\\
    \hline
    3 & 0.805 & 0.795 & 0.866 & 0.928\\
    \hline
    5 & 0.812 & 0.87 & 1.032 & 1.517\\
    \hline 
    8 & 0.83 & 7.231 & 3.751 & 4.43\\
    \hline
    9 & * & 17.659 & 74.469 & 89.852\\
    \hline
    11 & * & 39.295 & 665.628 & 1634.821\\
    \hline
  \end{tabular}
  \caption{$m$ represents the batch size, $d$ represents the feature num and * represents there are multiple batches that match the same gradient. The linear programming is solved by pulp library written in python with an intel i5-8500 CPU (3.00GHz) and the batch is sorted in alphabetical order to eliminate the impact of the sequence of the data in a batch.}
  \label{result}
\end{table}

In Section \ref{OneBatch} we use a linear programming method to obtain the whole batch data from the gradient computed, whereas this linear programming method has $d+d^2-d+2nd(d-1)=(2n+1)d^2-2nd$ constraints. As shown in Table \ref{result}, the more constraints the linear programming have, the longer time it takes to solve. Table.\ref{result} also shows with the increase of $m$ come multiple batches corresponding the same gradient even under the constriant that all elements are binary. 

\subsection{Batch gradient}
As discussed in Section \ref{Multi Batches}, \textit{Alice} cannot solve out \textit{Bob}'s data just based on the gradients. Thus avoiding leaking the batch gradient would be a effective method. The simplest way is to avoid using \textit{synchronized} method to compute the training gradients. Another method is to obscure the orginal gradient. For instance, \textit{Bob} may shuffle the sequence of his data to obfuscate \textit{Alice} or pushes the gradient selectively as proposed by \citet{shokri2015privacy}.



\section{Conclusion}
In this paper, we discuss the leakage in federated learning of approximated logistic regression model. We first showed how an honest-but-curious participant can easily infer the whole training data of other participants in synchronized case. We then quantified of the leakage in asychronized case and illustrated an honest-but-curious can infer nothing other than some certain constraints of the other participants' training data. We also analyzed how the hyperparameters of the learning, batch size and participant, affect the leakage and further proposed several plausible defenses. 


\bibliographystyle{unsrtnat}
\bibliography{ref}

\end{document}